\documentclass[]{interact}

\usepackage{epstopdf}
\usepackage[caption=false]{subfig}

\usepackage{amsmath,amssymb}
\usepackage{amsthm}
\usepackage{bm}
\usepackage{color}
\usepackage{graphics}
\usepackage{graphicx}
\usepackage{booktabs}
\usepackage{ascmac}
\usepackage{pb-diagram}
\usepackage{amscd} 
\usepackage{enumerate}
\usepackage[unicode,colorlinks,citecolor=blue]{hyperref}
\usepackage{here}

\newcommand{\lp}{\left(}  
\newcommand{\rp}{\right)}

\usepackage[longnamesfirst,sort]{natbib}
\bibpunct[, ]{(}{)}{;}{a}{,}{,}
\renewcommand\bibfont{\fontsize{10}{12}\selectfont}

\theoremstyle{plain}
\newtheorem{theorem}{Theorem}[section]
\newtheorem{lemma}[theorem]{Lemma}

\newtheorem{proposition}[theorem]{Proposition}

\newcommand{\indep}{\perp \!\!\! \perp}

\theoremstyle{definition}

\newtheorem{example}[theorem]{Example}

\theoremstyle{remark}

\begin{document}


\title{Verifiable identification condition for nonignorable nonresponse data with categorical instrumental variables}

\author{
\name{Kenji Beppu\textsuperscript{a}\thanks{Email: k.morikawa.es@osaka-u.ac.jp} and Kosuke Morikawa\textsuperscript{a}}
\affil{\textsuperscript{a}Graduate School of Engineering Science, Osaka University, Osaka, Japan}
}

\maketitle

\begin{abstract}
We consider a model identification problem in which an outcome variable contains nonignorable missing values. Statistical inference requires a guarantee of the model identifiability to obtain estimators enjoying theoretically reasonable properties such as consistency and asymptotic normality.
Recently, instrumental or shadow variables, combined with the completeness condition in the outcome model, have been highlighted to make a model identifiable. In this paper, we elucidate the relationship between the completeness condition and model identifiability when the instrumental variable is categorical. We first show that when both the outcome and instrumental variables are categorical, the two conditions are equivalent. However, when one of the outcome and instrumental variables is continuous, the completeness condition may not necessarily hold, even for simple models. Consequently, we provide a sufficient condition that guarantees the identifiability of models exhibiting a monotone-likelihood property, a condition particularly useful in instances where establishing the completeness condition poses significant challenges.
Using observed data, we demonstrate that the proposed conditions are easy to check for many practical models and outline their usefulness in numerical experiments and real data analysis.
\end{abstract}


\section{Introduction}
There has been a rapidly growing movement to utilize all the available data that may explicitly, even implicitly, contain missing values, such as causal inference \citep{ imbens2015causal} and data integration \citep{yang2020statistical, hu2022paradoxes}.
For such datasets, appropriate analysis of missing data is indispensable to correct selection bias owing to the missingness. In recent years, analysis of missing data under missing at random (MAR) assumption \citep{little2019statistical} has gradually matured \citep{robins1994estimation, kim2021statistical}. Although model identifiability is one of the most fundamental conditions in constructing the asymptotic theory, removing the MAR assumption makes statistical inference drastically difficult, especially in model identification \citep{miao2016identifiability}.
Estimation with unidentifiable models may provide multiple solutions that have exactly the same model fitting. Several researchers have considered giving sufficient conditions for the  model identification under missing not at random (MNAR).

Constructing observed likelihood consists of two distributions: (R) response mechanism and (O) outcome distribution \citep{kim2021statistical}. \cite{miao2016identifiability} considered identification condition with Logistic, Probit, and Robit (cumulative distribution function of $t$-distribution) models for (R) and normal and $t$ (mixture) distributions for (O).
\cite{cui2017identifiability} assumed Logistic, Probit, and cLog-log models for (R) and the generalized linear models for (O). These studies depend heavily on the model specification of both (R) and (O).
\cite{wang2014instrumental} introduced a covariate called instrument or shadow variable and demonstrated that the use of the instrument could considerably relax conditions on (R) and (O). For example, (O) requires only the monotone-likelihood property, which includes a variety of models, such as the generalized linear model.
\cite{tang2003analysis} and \cite{miao2018identification} derived conditions for model identifiability without postulating any assumptions on (R) with the help of the instrument. \cite{miao2019identification} further relaxed the assumption under an assumption referred to as the completeness condition on (R) \citep{d2010new, d2011completeness}.
For example, the generalized linear model with continuous covariates satisfies the completeness condition. To the best of our knowledge, this combination of an instrument on (R) and completeness on (O) is the most general condition for model identification and has been accepted in numerous studies \citep{zhao2022versatile, yang2019causal}.

Generally, assumptions on (O) rely on the distribution of the complete data, which is untestable from observed data. Recently, modeling (O') the observed or respondents’ outcome model, instead of (O), has been used to relax the subjective assumption \citep{miao2019identification, riddles2016propensity}. However, the observed likelihood with (R) and (O') involves an integration that makes the identification problem intractable.
\cite{morikawa2021semiparametric} and \cite{beppu2021} established that the integration can be computed explicitly with Logistic models for (R) and generalized linear models for (O') and derived identification condition.
For general response mechanisms and  respondents’ outcome distributions, the model identification remains an open question. Furthermore, when the instrument is categorical such as smoking history and sex, the completeness condition is not available. For example, \cite{ibrahim2001using} considered a study on the mental health of children in Connecticut and used the parents' report of the psychopathology of the child as the binary instrument.

In this paper, we consider an identification problem with an instrument for (R) and (O') that satisfies the monotone-likelihood ratio property.
Note that although our model setup is similar to \cite{wang2014instrumental}, we can check the validity of (O') with observed data, for example, by using the information criteria such as AIC and BIC.
Furthermore, we can use semiparametric/nonparametric methods for modeling both (O') and (R).

The rest of this paper is organized as follows. Section \ref{section2} introduces the notation and defines model identifiability.
Section \ref{section3} derives the proposed identification condition. We demonstrate the effects of identifiability via a limited numerical study in Section \ref{section4}.
Moreover, application to real data is presented in Section \ref{section5}. Finally, concluding remarks are summarized in Section \ref{section6}.
All the technical proofs are relegated to the Appendix.

\section{Basic setup} \label{section2}

\subsection{ Observed likelihood }

Let $\{ \bm{x}_i, y_i, \delta_i\}_{i=1}^n$ be independent and identically distributed samples from a distribution of $(\bm{x},y,\delta)$, where $\bm{x}$ is a fully observed covariate vector, $y$ is an outcome variable subject to missingness, and $\delta$ is a response indicator of $y$ being $1(0)$ if $y$ is observed (missing).
We use the generic notation $p(\cdot)$ and $p(\cdot\mid \cdot)$ for the marginal density and conditional density, respectively. For example, $p(\bm{x})$ is the marginal density of $\bm{x}$, and $p(y\mid \bm{x})$ is the conditional density of $y$ given $\bm{x}$.
We model the MNAR response mechanism $P(\delta=1\mid \bm{x},y)$ and consider its identification.
The observed likelihood is defined as 
\begin{align}
    \prod_{i: \delta_i=1} P(\delta_i=1\mid y_i,\bm{x}_i)p(y_i\mid \bm{x}_i) \prod_{i: \delta_i=0} \int \left\{ 1-P(\delta_i=1\mid y, \bm{x}_i) \right\} p(y\mid \bm{x}_i)dy. \label{1}
\end{align}
We say that this model is identifiable if parameters in (\ref{1}) are identified, which is equivalent to parameters in $P(\delta=1\mid y,\bm{x})p(y\mid \bm{x})$ being identified. This identification condition is essential even for semiparametric models such as an estimator defined by moment conditions \citep{morikawa2021semiparametric}. However, simple models can be easily unidentifiable.
For example, Example 1 in \cite{wang2014instrumental} presented an unidentifiable model when the outcome model is normal, and the response mechanism is a Logistic model.

There is an alternative way to express the relationship between $y$ and $\bm{x}$.
A disadvantage of modeling $p(y\mid \bm{x})$ is its subjective assumption on the distribution of complete data, not of observed data.
In other words, if we made assumptions about $p(y\mid \bm{x})$ and ensured its identifiability, we could not verify the assumptions using the observed data.
By contrast, this issue can be overcome by modeling $p(y\mid \bm{x}, \delta=1)$ because $p(y\mid \bm{x}, \delta=1)$ is the outcome model for the observed data, and we can check its validity using ordinal information criteria such as AIC and BIC.
Therefore, we model $p(y\mid \bm{x}, \delta=1)$ and consider the identification condition in Section \ref{section3}. Hereafter, we assume two parametric models $p(y\mid \bm{x},\delta=1;\bm{\gamma})$ and $P(\delta=1\mid \bm{x},y;\bm{\phi})$, where $\bm{\gamma}$ and $\bm{\phi}$ are parameters of the outcome and response models, respectively.  Although our method requires two parametric models, the class of identifiable models is very large. For example, it can include semiparametric outcome models for $p(y\mid \bm{x}, \delta=1;\bm{\gamma})$ and general response models $P(\delta=1\mid \bm{x},y;\bm{\phi})$ other than Logistic models, as discussed in Example \ref{satisfying_C3}.

\subsection{ Estimation } \label{estimation}

We present a procedure of parameter estimation based on parametric models of $p(y\mid \bm{x}, \delta=1;\bm{\gamma})$ and $P(\delta=1\mid \bm{x},y;\bm{\phi})$.
Let $\hat{\bm{\gamma}}$ be the maximum likelihood estimator of $\bm{\gamma}$. The observed likelihood (\ref{1}) yields to the mean score equation for $\bm{\phi}$ \citep{kim2021statistical}:
\begin{align*}
    \sum_{i=1}^{n} \left\{ \delta_i \frac{\partial \log \pi(\bm{x}_i,y_i;\bm{\phi})}{\bm{\phi}} - (1-\delta_i) \frac{ \int \partial \pi(\bm{x}_i,y;\bm{\phi})/\partial\bm{\phi} \cdot p(y\mid \bm{x}) dy }{ \int \{ 1- \pi(\bm{x}_i,y;\bm{\phi})\} p(y\mid \bm{x}) dy } \right\}=0
\end{align*}
where $\pi(\bm{x},y;\bm{\phi})=P(\delta=1\mid \bm{x},y;\bm{\phi})$. By using Bayes' formula $ p(y\mid \bm{x}) \propto p(y\mid \bm{x}, \delta=1)/\pi(\bm{x},y;\bm{\phi})$, the mean score can be written as 
\begin{align*}
    \sum_{i=1}^{n} \left\{ \delta_i s_1( \bm{x}_i, y_i; \bm{\phi} ) + (1-\delta_i) s_0( \bm{x}_i;\bm{\phi} ) \right\}=0,
\end{align*}
where
\begin{align*}
    s_1( \bm{x}, y; \bm{\phi} )= \frac{ \partial \log \pi(\bm{x},y;\bm{\phi}) }{ \partial\bm{\phi} }, \ s_0( \bm{x};\bm{\phi} )= - \frac{ \int s_1( \bm{x}, y; \bm{\phi} ) p(y\mid \bm{x}, \delta=1) dy }{ \int \left\{  1/\pi(\bm{x},y;\bm{\phi}) -1 \right\} p(y\mid \bm{x}, \delta=1) dy  }.
\end{align*}
To compute the two integrations in $s_0(\cdot)$, we can use the fractional imputation \citep{kim2011parametric}.
As described in \cite{riddles2016propensity}, the EM algorithm is also applicable.

\section{Identifiability} \label{section3}

\subsection{Definition of identification}
Recall that the identification condition in (\ref{1}) is for parameters in $P(\delta=1\mid y,\bm{x}) p(y\mid \bm{x})$.
As seen in Section \ref{estimation}, the conditional density $p(y\mid \bm{x})$ is represented by $p(y\mid \bm{x}, \delta=1;\bm{\gamma})$ and $P(\delta=1\mid \bm{x},y;\bm{\alpha},\bm{\phi})$ by Bayes' formula. Thus, using the formula, identification with these models changes to parameters in $\varphi(y,\bm{x};\bm{\phi},\bm{\gamma})$, where
\begin{align}
   \varphi(y,\bm{x};\bm{\phi},\bm{\gamma})= \frac{ p(y\mid \bm{x}, \delta=1;\bm{\gamma}) }{\int p(y\mid \bm{x}, \delta=1;\bm{\gamma}) / \pi(\bm{x},y;\bm{\phi}) dy }. \label{integral}
\end{align}
Strictly speaking, the identification condition is $\varphi(y,\bm{x};\bm{\phi},\bm{\gamma})=\varphi(y,\bm{x};\bm{\phi}^{\prime},\bm{\gamma}^{\prime})$ with probability $1$ implies that $(\bm{\phi}^{\top},\bm{\gamma}^{\top})=({\bm{\phi}^{\prime}}^{\top},{\bm{\gamma}^{\prime}}^{\top})$.
Generally, the integral in the denominator of (\ref{integral}) does not have the closed form, which makes deriving a sufficient condition for the identifiability quite challenging.
\cite{morikawa2021semiparametric} identified a combination of Logistic models and normal distributions for response and outcome models has a closed form of the integration and derived a sufficient condition for the model identifiability.
\cite{beppu2021} extended the model to a case where the outcome model belongs to the exponential family while the response model is still a Logistic model.
However, when the response mechanism is general, simple outcome models such as normal distribution can be unidentifiable.

\begin{example} \label{unidentifiable model}
Suppose that the respondents' outcome model is $y\mid (\delta=1,x)\sim N(\gamma_0+\gamma_1x,1)$, and the response model is $P(\delta=1\mid x,y)=\Psi(\alpha_0+\alpha_1x+\beta y)$, where $\Psi$ is a known distribution function such that the integration in (\ref{integral}) exists; then, this model is unidentifiable. For example, different parametrization 
$(\alpha_0,\alpha_1,\beta,\gamma_0,\gamma_1)=(0,1,1,0,1)$, $(\alpha_0^{\prime},\alpha_1^{\prime},\beta^{\prime},\gamma_0^{\prime},\gamma_1^{\prime})=(0,3,-1,0,1)$ yields the same value of the observed likelihood. 
\end{example}

Recently, widely applicable sufficient conditions have been proposed. Assume that a covariate $\bm{x}$ has two components, $\bm{x}=(\bm{u}^{\top},\bm{z}^{\top})^{\top}$, such that
\begin{itemize}
\item [(C1)]
    $\bm{z} \indep \delta \mid (\bm{u},y)$ and $  \bm{z} \not\indep y\mid (\delta=1,\bm{u}).$
\end{itemize}
The covariate $\bm{z}$ is called an instrument \citep{d2010new} or a shadow variable \citep{miao2016varieties}. \cite{miao2019identification} derived sufficient conditions for model identifiability by combining the instrument and the completeness condition:
\begin{itemize}
    \item[(C2)] For all square-integrable function $h(\bm{u},y)$, $E[h(\bm{u},y)\mid \delta=1,\bm{u},\bm{z}]=0$ almost surely implies $h(\bm{u},y)=0$ almost surely.
\end{itemize}

\begin{lemma}[Identification condition by \cite{miao2019identification}]\label{miaoresult}
Under the conditions (C1) and (C2), the joint distribution $p(y,\bm{u},\bm{z},\delta)$ is identifiable.
\end{lemma}

Although the completeness condition is useful and applicable for general models, a simple model with a categorical instrument does not hold the completeness condition.

\begin{example}[Violating completeness with categorical instrument]\label{counter example}
Suppose $y\mid (\delta=1, u,z)$ follows the normal distribution $N(u+z, 1)$, and an instrument $z$ is binary taking $0$ or $1$.
This distribution does not satisfy the completeness condition because the conditional expectation $E[h(u,y)\mid \delta=1,u,z]=0$ when $h(u,y)=1+y-u-(y-u)^2$.
\end{example}

A vital implication of Example \ref{counter example} is that instruments are no longer evidence of model identification when the instrument is categorical.
Developing the identification condition for models with discrete instruments is important in applications \citep{ibrahim2001using}.
We separately discuss two cases: (i) both $y$ and $\bm{z}$ are categorical; (ii) respondents’ outcome model has the monotone-likelihood ratio property.

When all variables, $y$ and $\bm{z}$, are categorical, the model can be fully nonparametric. Theorem \ref{full nonpara} demonstrates that, under these conditions, the completeness and identifiability conditions are equivalent. See Appendix 2 in \cite{riddles2016propensity} for the estimation of such fully nonparametric models.

\begin{theorem} \label{full nonpara}
When both $y$ and $\bm{z}$ are categorical, under condition (C1), the joint distribution $p(y,\bm{u},\bm{z},\delta)$ is identifiable if and only if condition (C2) holds.
\end{theorem}

As evidenced in Lemma \ref{miaoresult}, condition (C2) is generally sufficient for model identifiability, but Theorem \ref{full nonpara} also reveals that it is necessary when $y$ and $\bm{z}$ are categorical.

Next, we consider the identification condition for the other case (ii). Let $\mathcal{S}_{y}$ be the support of the random variable $y$. We assume the following four conditions:

\begin{itemize}
    \item[(C3)] The response mechanism is
    \begin{align}
        P(\delta=1\mid y, \bm{x};\bm{\phi})=P(\delta=1\mid y, \bm{u};\bm{\phi})=\Psi\{ h(\bm{u};\bm{\alpha}) + g(\bm{u};\bm{\beta}) m(y) \}, \label{response}
    \end{align}
    where $\bm{\phi}=(\bm{\alpha}^{\top}, \bm{\beta})^{\top}$, 
    $m:\mathcal{S}_y\to\mathbb{R}$ and $\Psi:\mathbb{R}\to(0,1]$ are known continuous strictly monotone functions, and $h(\bm{u};\bm{\alpha})$ and $g(\bm{u};\bm{\beta})$ are known injective functions of $\bm{\alpha}$ and $\bm{\beta}$, respectively.

    \item[(C4)] The density or mass function $p(y\mid \bm{x}, \delta=1;\bm{\gamma})$ is identifiable, and its support does not depend on $\bm{x}$.

    \item[(C5)] For all $\bm{u}\in \mathcal{S}_{\bm{u}}$, there exist $\bm{z}_1$ and $\bm{z}_2$, such that $p(y\mid \bm{u},\bm{z}_1, \delta=1)\neq p(y\mid \bm{u},\bm{z}_2, \delta=1)$, and $p(y\mid \bm{u},\bm{z}_1, \delta=1)/ p(y\mid \bm{u},\bm{z}_2, \delta=1)$ is monotone.

    \item[(C6)] 
    \begin{align*}
        \int \frac{ p(y\mid \bm{x}, \delta=1;\bm{\gamma})}{ \Psi\{ h(\bm{u};\bm{\alpha}) + g(\bm{u};\bm{\beta}) m(y) \}} dy <\infty \ \ \ \mathrm{a.s.}
    \end{align*}
\end{itemize}

The condition (C3) means that the random variable $\bm{z}$ plays a role of an instrument. The condition (C4) is the identifiability of $p(y\mid \bm{x}, \delta=1;\bm{\gamma})$, which is testable from the observed data.
The condition (C5) assumes a monotone-likelihood property on the outcome model, which was also used in \cite{wang2014instrumental} for the complete data. 
The condition (C6) is necessary for (\ref{1}) to be well-defined. It is essentially the same condition as Theorem 3.1 (I1) of \cite{morikawa2021semiparametric}.
This condition is always true when the support of $y$ is finite.
However, it must be carefully verified when $y$ is continuous. See Proposition \ref{cor_normal} below for useful sufficient conditions when the respondents' outcome model is normal distribution.

Under conditions (C3)--(C6), we obtain the desired identification condition.

\begin{theorem} \label{main}
The parameter $(\bm{\phi}^{\top}, \bm{\gamma}^{\top})^{\top}$ is identifiable if the conditions (C1) and (C3)--(C6) hold.
\end{theorem}

We provide an example of outcome models satisfying the condition (C5).

\begin{example}[Model satisfying (C5)] \label{exponential family}
Let density functions in the exponential family be
\begin{align*}
    p(y\mid \bm{x}, \delta=1;\bm{\gamma})=\exp \left( \frac{y\theta -b(\theta)}{\tau} +c(y;\tau) \right),
\end{align*}
where $\theta=\theta(\eta)$, $\eta=\sum_{l=1}^{L} \eta_l(\bm{x})\kappa_l$, $\bm{\kappa}=(\kappa_1,\ldots,\kappa_L)^{\top}$, and $\bm{\gamma}=(\tau,\bm{\kappa}^{\top})^{\top}$. Then the density ratio becomes
\begin{align*}
    \frac{p(y\mid \bm{u},\bm{z}_1, \delta=1)}{p(y\mid \bm{u},\bm{z}_2, \delta=1)}\propto \exp \left( \frac{\theta_1-\theta_2}{\tau} y \right),
\end{align*}
where $\bm{x}_i=(\bm{u},\bm{z}_i)$ and $\theta_i=\theta\{\sum_{l=1}^{L} \eta_l(\bm{x}_i)\kappa_l\},\ i=1,2$.
Therefore, the density ratio is monotone.
\end{example}

\begin{example}[Model satisfying (C6)] \label{satisfying_C3}
In application, it is often reasonable to assume a normal distribution on the respondents’ outcome model. Focusing on the tail of the outcome model, we provide a sufficient condition to check (C6) for models with general response mechanisms.
\end{example}

\begin{proposition}\label{cor_normal}
Suppose that the observed distribution $p(y\mid \bm{x}, \delta=1)$ is normal distribution $N(\mu(\bm{x};\bm{\kappa}),\sigma^2)$, the response mechanism is (\ref{response}) with $m(y)=y$ and $g(\bm{u};\bm{\beta})=\beta$, and the strictly monotone increasing function $\Psi$ meets the following condition:
\begin{align}
    ^{\exists}s\in (0,2)\ \mathrm{s.t.} \ \liminf_{z\to-\infty} \Psi(z) \exp(|z|^s)>0 \label{additional condition}.
\end{align}
Then, this model satisfies (C6).
\end{proposition}

The condition (\ref{additional condition}) is easy to check. For example, it holds for Logistic and Robit functions but not for the Probit function.
According to Proposition \ref{cor_normal}, it is possible to estimate $\mu(x;\bm{\kappa})$ with observed data using splines and other nonparametric methods, which allows us to use very flexible models.
Furthermore, we can also estimate the response mechanism using nonparametric methods because it does not impose any restrictions on the functional form of $h(\bm{u};\bm{\alpha})$.

\section{Numerical experiment} \label{section4}
We present the effects of identifiability in numerical experiments by comparing weak and strong identifiable models. We prepared four Scenarios S1--S4:

\begin{itemize}
    \item[S1:] (Outcome: Normal, Response: Logistic)\\
    $[y\mid u,z,\delta=1]\sim N( \kappa_0+\kappa_1u+\kappa_2z,\sigma^2)$, $\text{logit}\{ P(\delta=1\mid u,y;\bm{\alpha},\beta)\}=\alpha_0+\alpha_1u+\beta y$, $u\sim N(0,1^2)$, and $z\sim B(1,0.5)$, where $(\kappa_0,\kappa_1,\sigma^2)^{\top}=( 0.3,0.4,1/{\sqrt{2}^2} )^{\top}$ and $(\alpha_0,\alpha_1,\beta)^{\top}=( 0.7,-0.2,0.29 )^{\top}$.
    \item[S2:] (Outcome: Normal, Response: Cauchy)\\
     $[y\mid u,z,\delta=1] \sim N(\kappa_0+\kappa_1u+\kappa_2z,\sigma^2 )$, $P(\delta=1\mid u,y;\bm{\alpha},\beta)=\Psi( \alpha_0+\alpha_1u+\beta y)$, $u\sim \mathrm{Unif}(-1,1)$, and $z\sim B(1,0.7)$, where $(\kappa_0,\kappa_1,\sigma^2)^{\top}=( -0.36,0.59,1/{\sqrt{2}^2} )^{\top}$, $(\alpha_0,\alpha_1,\beta)^{\top}=( 0.24,-0.1,0.42)^{\top}$, and $\Psi$ is the cumulative distribution function of the Cauchy distribution.
    \item[S3:] (Outcome: Bernoulli, Response: Probit)\\
    $[y\mid u,z,\delta=1] \sim B(1,p(u,z;\bm{\kappa})\} )$, $P(\delta=1\mid u,y;\bm{\alpha},\beta)=\Psi( \alpha_0+\alpha_1u+\beta y)$, $u\sim N(0,1^2)$, and $z\sim N(0,1^2)$, where $p(u,z;\bm{\kappa})=1/\{ 1+\exp( -\kappa_0-\kappa_1u -\kappa_2z )$, $(\kappa_0,\kappa_1,\kappa_2)^{\top}=( -0.21, 3.8, 1.0 )^{\top}$, $(\alpha_0,\alpha_1,\beta)^{\top}=( 0.4,0.39,0.3)^{\top}$, and $\Psi$ is the cumulative distribution function of the standard normal.
     \item[S4:] (Outcome: Normal+nonlinear mean structure, Response: Cauchy or Logistic)\\
     $[y\mid u,z,\delta=1]\sim N( \mu(\bm{x}),0.5^2 )$, $P(\delta=1\mid u,y;\bm{\alpha},\beta)=\Psi( \alpha_0+\alpha_1u+\beta y)$, $u\sim \mathrm{Unif}(-1,1)$, and $z\sim B(1,0.5)$, where $\mu(\bm{x})=z+\cos(2\pi u)+\exp(z+u)$, $(\alpha_0,\alpha_1,\beta)^{\top}=( 0.1,-0.2,0.3)^{\top}$, and $\Psi$ is the cumulative distribution function of the Cauchy or Logistic distribution. 
\end{itemize}

In S1 and S2, the strength of the identification can be adjusted by changing the parameter $\kappa_2$ because $\kappa_2=0$ indicates that the model is unidentifiable by Example \ref{unidentifiable model}.
On the other hand, we can verify that the models in S3 and S4 are identifiable by Theorem \ref{main}. For example, in S4, we can see that checking (C3) and (C4) is straightforward to the setting, while (C5) and (C6) hold from Example \ref{exponential family} and Proposition \ref{cor_normal}, respectively.
From S3 and S4, we can confirm the successful inference even in the case of discrete outcome and complex mean structures, respectively.

We generated 1,000 independent Monte Carlo samples and computed two estimators for $E[y]$ and $\beta$ with two methods: fractional imputation (FI) and complete case (CC) estimators, which use only completely observed data.
The estimator for $E[y]$ is computed by the standard inverse probability weighting method with estimated response models \citep{riddles2016propensity}. We used correctly specified models for Scenarios S1--S3 but used nonparametric models for Scenario S4 because it is unrealistic to assume that the complicated mean structure is known. The R package `crs' specialized in nonparametric spline regression on the mixture of categorical and continuous covariates \citep{nie2012crs} is used to estimate the respondents' outcome model.
Response models are estimated by using the method discussed in Section \ref{estimation}.

Bias, root mean squared error (RMSE), and coverage rate for 95\% confidence intervals in S1--S4 are reported in Table \ref{normal simulation}.
In all the Scenarios, CC estimators have a significant bias, and the coverage rates are far from 95\%, while FI estimators work well when the model is surely identifiable. When $\kappa_2$ is small in S1 and S2, the performance of variance estimation with FI is poor, as expected, although that of point estimates is acceptable. The results in S4 indicate that the model is identifiable even if we use a nonparametric mean structure, and the estimates are almost the same between the two response models.

\begin{table}[H] 
\centering
  \caption{Results of S1--S4: Bias, root mean square error (RMSE), and coverage rate (CR,\%) with $95\%$ confidence interval are reported.
  CC: complete case; FI: fractional imputation.}
  \label{normal simulation}
  \begin{tabular}{ccccrrr}  
  \cmidrule(lr){1-7}
     Scenario &Parameter & $\kappa_2$ & Method & Bias & RMSE & CR \\ 
     \cmidrule(lr){1-7}
     & &\raisebox{-1.5ex}[0ex][-2ex]{1.0} & CC & 0.053 &  0.066 & 73.5 \\ 
      &  & & FI &  0.000 & 0.043 & 95.4 \\
      &\raisebox{-1.5ex}[0ex][-1.5ex]{$E[y]$}  &\raisebox{-1.5ex}[0ex][-1ex]{0.5} & CC & 0.039 & 0.053 & 80.9 \\ 
      &  & & FI &  -0.001 & 0.059 & 97.1 \\ 
     S1  &  &\raisebox{-1.5ex}[0ex][-1ex]{0.1} & CC & 0.034 & 0.049 & 83.0 \\ 
      &  & & FI &  0.021 & 0.136 & 99.8 \\ 
      \cmidrule(lr){2-7}
      &   &1.0 & FI &  0.001 & 0.163 & 95.2 \\   
      &$\beta$   &0.5  & FI &  0.003 & 0.330 & 98.6 \\   
      &   &0.1 & FI &  -0.146 & 0.865 & 100 \\   
    \cmidrule(lr){1-7}
     & &\raisebox{-1.5ex}[0ex][-2ex]{1.0} & CC & 0.146 & 0.152 & 5.7 \\ 
      &  & & FI &  -0.004 & 0.051 & 94.8 \\
      &\raisebox{-1.5ex}[0ex][-1.5ex]{$E[y]$}  &\raisebox{-1.5ex}[0ex][-1ex]{0.5} & CC & 0.130 & 0.136 & 7.7 \\ 
      &  & & FI &  -0.008 & 0.086 & 86.2 \\ 
     S2  &  &\raisebox{-1.5ex}[0ex][-1ex]{0.1} & CC & 0.127 & 0.133 & 9.4 \\ 
      &  & & FI &  -0.007 & 0.105 & 92.4 \\ 
      \cmidrule(lr){2-7}
      &   &1.0 & FI &  0.008 & 0.148 & 95.4 \\   
      &$\beta$   &0.5  & FI &  0.044 & 0.365 & 100 \\   
      &   &0.1 & FI &  0.033 & 0.448 & 100 \\  
     \cmidrule(lr){1-7}
    &  \raisebox{-1.5ex}[0ex][-1ex]{$E[y]$}&-- & CC & 0.100 & 0.102 & 0.3 \\ 
   S3    &  &-- &FI & 0.001 & 0.022 & 95.3 \\ 
   \cmidrule(lr){2-7}
     & $\beta$ &--& FI &  -0.023 & 0.279 & 95.0 \\ 
      \cmidrule(lr){1-7}
    &  &-- & CC(Logistic) & 0.341 & 0.355 & 5.4 \\ 
       &  \raisebox{-1.5ex}[0ex][-1ex]{$E[y]$}&-- &FI(Logistic) & 0.005 & 0.079 & 95.4 \\ 
       &  &-- & CC(Cauchy) & 0.296 & 0.312 & 10.7 \\ 
   S4    &  &-- &FI(Cauchy) & 0.007 & 0.080  & 94.3 \\
   \cmidrule(lr){2-7}
     & \raisebox{-1.5ex}[0ex][-1ex]{$\beta$} &--& FI(Logistic) &  0.006 & 0.050 & 94.7 \\ 
     &  &--& FI(Cauchy) & 0.011  & 0.063 & 93.8 \\ 
     \cmidrule(lr){1-7}
  \end{tabular}
\end{table}

\section{Real data analysis} \label{section5}
We analyzed a dataset of 2139 HIV-positive patients enrolled in AIDS Clinical Trials Group Study 175 (ACTG175; \cite{hammer1996trial}).
In this analysis, we specify 532 patients for analysis who received zidovudine (ZDV) monotherapy.
Let each $y$, $x_1$, and $x_2$ be the CD4 cell count at $96\pm 5$ weeks, at the baseline, and at $20\pm 5$ weeks, $x_3$ be the CD8 cell count at the baseline, and $z$ be sex.
The outcome was subject to missingness with a 60.34\% observation rate, while all covariates were observed.
To make estimation stable and easy, we standardized all the data.
We expect that $z$ (sex) is a reasonable choice for an instrument variable because the information is a biological value, which affects the value of CD4, but has little effect on the response probability.

Patients who are suffering from a mild illness of HIV tend to have higher CD4 cell count; thus, one may consider that missingness of the outcome relates to serious conditions and may expect that the missing value of the outcome would be a lower CD4 cell count than the respondent. We therefore considered five different MNAR response models:
\begin{align*}
    P(\delta=1\mid x_1,x_2,x_3,y)=\Psi( \alpha_0+\alpha_1x_1+\alpha_2x_2+\alpha_3x_3+\beta y ),
\end{align*}
where $\Psi$ represents either the Logistic function or the distribution functions of the Cauchy or $t$ distribution with degrees of freedom $v\,(=2,5,10)$. Theorem \ref{main} and Proposition \ref{cor_normal} ensure that all the models with these five response models are identifiable, even when the instrumental variable $z$ is discrete.
From the above conjecture on missing values, the sign of $\beta$ is expected to be negative.
We assumed that the respondent's outcome is a normal distribution with a nonparametric mean structure and estimated by the `crs' R package as considered in Scenario S4 in Section \ref{section4}. The residual plots shown in Figure \ref{residual plot} and the computed $R^2$-value $(=0.453)$ signify the assumed distribution on the respondents' outcome fit well.
Table \ref{real data result} reports the estimated parameters and their estimated standard errors calculated by 1,000 bootstrap samples. The results of the five response models were almost similar. This suggests that the response mechanism is robust to the choice of response models.
Although we cannot determine whether it is MNAR or MAR because the estimated standard error for $\beta$ is large, the point estimate is negative, as we expected.
This result is consistent with the result in \cite{zhao2021sufficient}.

\begin{figure}[H]
    \centering
    \includegraphics[keepaspectratio, scale=0.5]{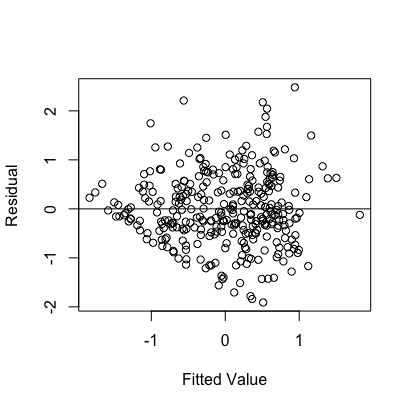}
    \caption{Residual plots of respondents' outcome in ACTG175 data.\label{residual plot}}
\end{figure}

\begin{table}[H] 
\centering
  \caption{Estimated parameters: Estimates and standard errors for the target parameters are reported. Logistic and Cauchy are Fractional Imputation using Logistic and Cauchy distributions for the response mechanism. $T_v$: $t$ distribution function with degrees of freedom $v\,(=2,5,10)$.\label{real data result}}
  \begin{tabular}{cccc|cccc}  
  \cmidrule(lr){1-8}
      Parameter   & Model & Estimate & SE & Parameter   & Model & Estimate & SE \\ 
     \cmidrule(lr){1-8}
        & Logistic & 0.464  & 0.104 & & Logistic& 0.125 & 0.156 \\
        & Cauchy & 0.417 & 0.260 &  &  Cauchy & 0.108 & 0.139 \\ 
      $\alpha_0$  & $T_2$ & 0.341 & 0.081 & $\alpha_1$  & $T_2$ & 0.091 & 0.113  \\
        & $T_5$ & 0.306 & 0.069 & & $T_5$ & 0.082 & 0.102 \\
        & $T_{10}$ & 0.295 & 0.066& & $T_{10}$ & 0.080 & 0.099 \\
      \cmidrule(lr){2-4}
      \cmidrule(lr){6-8}
         & Logistic & 0.255 & 0.192 & &Logistic & 0.093 & 0.107 \\
        & Cauchy & 0.244 & 0.207 &  & Cauchy & 0.083 & 0.097 \\ 
      $\alpha_2$  & $T_2$ & 0.196 & 0.148 & $\alpha_3$  & $T_2$ & 0.069 & 0.079  \\
        & $T_5$ & 0.169 & 0.126 & & $T_5$ & 0.062 & 0.070 \\ 
        & $T_{10}$ & 0.160 & 0.120& & $T_{10}$ & 0.060 & 0.068 \\
     \cmidrule(lr){2-4}
      \cmidrule(lr){6-8}
        & Logistic & -0.032 & 0.314 & &Logistic & 276.70 & 13.476 \\
        & Cauchy & -0.030 & 0.387 &  & Cauchy & 276.51 & 14.107 \\ 
      $\beta$  & $T_2$ & -0.027 & 0.235 & $E[y]$  & $T_2$ & 276.57 & 13.437  \\
        & $T_5$ & -0.021 & 0.203 & & $T_5$ & 276.61 & 13.271 \\ 
        & $T_{10}$ & -0.019 & 0.194& & $T_{10}$ & 276.63 & 13.217 \\
      \cmidrule(lr){1-8}
  \end{tabular}
\end{table}

\section{Conclusion} \label{section6}

In this paper, we proposed a new identification condition for models using respondents’ outcome and response models. Although our method requires the specification of the two models, the model can be very general with the help of an instrument.
As considered in Scenario S4 in Section \ref{section4}, the mean function in the  respondents' outcome model can be nonparametric, and the response model can be any strictly monotone function, other than Logistic models.
Our condition guarantees model identifiability even when instruments are categorical, which is not covered by previous conditions. Another advantage of using our method is the identification condition is easy to verify with observed data.

However, our method has some limitations. First, respondents' outcome models need to have the monotone-likelihood property by Condition (C5). For example, we cannot deal with mixture models in our framework. Second, the specification of instruments is necessary in advance. To date, some studies on finding the instruments have been proposed \citep{zhao2021sufficient}, but there are still no gold standard methods.

\section*{Funding}
Research by the second author was supported by MEXT Project for Seismology toward Research Innovation with Data of Earthquake (STAR-E) Grant Number JPJ010217.

\bibliography{bibliography.bib}


\appendix

\section{Technical Proofs} \label{Appendix A}

We first provide a technical result to prove Theorem \ref{full nonpara}.
\begin{lemma} \label{appendix_lem}
Let $a$, $b$, and $c$ be any positive real numbers. Assume that $r_1$ and $r_2$ are positive real numbers satisfying
\begin{align}
- \frac{ ab}{a+b}<\frac{r_1^2 - r_2^2}{r_1^2r_2^2} <c . \label{cond_abc}
\end{align}
Then, there exist $0<\pi^{(k)}_j<1\,(j=1,2,3; k=1,2)$ such that
\begin{align}
\sum_{j=1}^3 \pi^{(1)}_j = r_1^2,\quad \sum_{j=1}^3 \pi^{(2)}_j = r_2^2, \label{pi_scale}
\end{align}
and 
\begin{align}
\begin{split}
\frac{1}{\pi^{(1)}_1} - \frac{1}{\pi^{(2)}_1} = a, \quad  
\frac{1}{\pi^{(1)}_2} - \frac{1}{\pi^{(2)}_2} = b, \quad 
\frac{1}{\pi^{(1)}_3} - \frac{1}{\pi^{(2)}_3} = -c.
\end{split} \label{const_pi}
\end{align}
\end{lemma}
\begin{proof}[Proof of Lemma \ref{appendix_lem}]
By using a polar coordinate system, we transform $\pi^{(k)}_j\,(j=1,2,3;k=1,2)$ into
\begin{align*}
(\sqrt{\pi^{(1)}_1},\sqrt{\pi^{(1)}_2},\sqrt{\pi^{(1)}_3)}&=r_1(\sin \phi_1\cos \phi_2, \sin \phi_1\sin\phi_2, \cos \phi_1 ),\\
(\sqrt{\pi^{(2)}_1},\sqrt{\pi^{(2)}_2},\sqrt{\pi^{(2)}_3})&=r_2(\sin \psi_1\cos \psi_2, \sin\psi_1 \sin \psi_2, \cos\psi_1 ),
\end{align*}
where $0<\phi_1,\phi_2,\psi_1,\psi_2<\pi/2$ to ensure $\pi^{(k)}_j\,(j=1,2,3;k=1,2)$ satisfy \eqref{pi_scale}. It follows from \eqref{const_pi} and double-angular formulas that we have
\begin{align}
\begin{split}
&r_1^2(1-\omega_1)(1+\omega_2) - r_2^2(1-\omega_3)(1+\omega_4)\\
&= -\frac{ ar_1^2r_2^2}{4} (1-\omega_1)(1+\omega_2)(1-\omega_3)(1+\omega_4),
\end{split}\label{form1}\\
\begin{split}
&r_1^2(1-\omega_1)(1-\omega_2) - r_2^2(1-\omega_3)(1-\omega_4)\\
&= -\frac{ br_1^2r_2^2}{4} (1-\omega_1)(1-\omega_2)(1-\omega_3)(1-\omega_4),
\end{split}\label{form2} \\
&r_1^2(1+\omega_1) - r_2^2(1+\omega_3) = \frac{ cr_1^2r_2^2}{2} (1+\omega_1)(1+\omega_3),\label{form3}
\end{align}
where $\omega_1=\cos 2\phi_1,\omega_2=\cos2\phi_2,\omega_3=\cos 2\psi_1$, and $\omega_4=\cos2\psi_2$. Setting $\omega_2=\omega_4$ and equations \eqref{form1} and \eqref{form2} yield
\begin{align*}
r_1^2(1-\omega_1) - r_2^2(1-\omega_3)&= -\frac{ ar_1^2r_2^2}{4} (1-\omega_1)(1+\omega_2)(1-\omega_3), \\
r_1^2(1-\omega_1) - r_2^2(1-\omega_3) &= -\frac{ br_1^2r_2^2}{4} (1-\omega_1)(1-\omega_2)(1-\omega_3).
\end{align*}
Fixing $\omega_2=1-2a/(a+b)$ reduces the above equations to the one common equation
\begin{align}
r_1^2(1-\omega_1) - r_2^2(1-\omega_3) = -\frac{ r_1^2r_2^2ab}{2(a+b)} (1-\omega_1)(1-\omega_3), \label{form4}
\end{align}
maintaing the condition $-1<\omega_2<1$.
It remains to show that there exists $-1<\omega_3<1$ satisfying \eqref{form3} and \eqref{form4}. Solving the equation \eqref{form4} with respect to $\omega_1$, we have
\begin{align}
\omega_1=1- \frac{r_2^2(1-\omega_3)}{r_1^2 +  r_1^2r_2^2ab(1-\omega_3)/\{ 2(a+b) \}}.\label{form5}
\end{align}
Substituting \eqref{form5} into \eqref{form3} leads to the following quadratic equation with respect to $\omega_3$:
\begin{align*}
f(\omega_3)=&\lp \frac{  r_1^2r_2^4 ab +  c r_1^4r_2^4 ab }{2(a+b)} - \frac{ c r_1^2r_2^4}{2} \rp \omega_3^2 
-\lp \frac{  r_1^4 r_2^2 a b }{a+b} + c r_1^4r_2^2  \rp \omega_3 \\
&+\lp \frac{  r_1^2r_2^2 ab  \lp 2 r_1^2  -r_2^2  - c r_1^2 r_2^2 \rp }{2(a+b)} +\frac{  c r_1^2r_2^4 }{2}+ 2r_1^4 - 2r_1^2r_2^2 -  c r_1-4r_2^2  \rp=0.
\end{align*}
It follows from \eqref{cond_abc} that 
\begin{align*}
f(1)=r_1^2\lp 2r_1^2 - 2r_2^2-2  cr_1^2r_2^2\rp<0,\quad 
f(-1)=2r_1^2 \lp  r_1^2- r_2^2 + \frac{ r_1^2r_2^2ab}{ a+b}  \rp>0,
\end{align*}
which implies that there is at least one solution of $\omega_3$ to the equation $f(\omega_3)=0$ in the open interval $(-1,1)$ .

\end{proof}

Finally, we prove Theorem \ref{full nonpara} with the help of Lemma \ref{appendix_lem}.

\begin{proof}[Proof of Theorem \ref{full nonpara}]
Without loss of generality, we set the value of $\bm{u}$ to a fixed vector because the following proof holds for each $\bm{u}$. 
Let the categorical variables $y$ and $z$ take values in $\{ 1,2,\ldots, p \}$ and $\{ 1,2,\ldots, q \}$, respectively. We show that model identifiability implies the completeness condition (C2) by individually addressing three cases: (i) $p=2$, (ii) $p=3$, and (iii) $p\geq 4$ because ``if" part has been already established by Lemma \ref{miaoresult}.

When $p=2$, condition (C1) results in the rank of a $q \times 2$ matrix, composed of $p(y=j\mid \delta=1,z=i)$ in its $(i,j)$-th element ($i=1,2; j=1\dots, q$), being $2$. Hence, identifiable models always satisfy the completeness condition (C2).

For cases where $p\geq 3$, we must show that the model becomes unidentifiable when the completeness condition is violated. The breach of the completeness condition indicates the existence of a non-zero vector $(h_1, \dots, h_p)$ such that for $z=1,\ldots,q$, we have
\begin{align}
E[h_y\mid \delta=1,z]=\sum_{y=1}^{p} h_y p(y\mid \delta=1,z)=0. \label{complete_cond}
\end{align}
The elements in $(h_1,\cdots,h_p)$ do not all share the same sign, and multiplying this vector by any constant does not affect the above equation. Recall that the model's unidentifiability implies that $\pi^{(1)}_y\neq \pi^{(2)}_y$ exists for some $y\in\{1,\dots, p\}$, satisfying $\sum_{y=1}^{p} p(y\mid \delta=1,z)/\pi^{(1)}_y=\sum_{y=1}^{p} p(y\mid \delta=1,z)/\pi^{(2)}_y$. We now construct an unidentifiable model when the completeness condition is violated.

When $p=3$, without loss of generality, we assume $h_1>0$, $h_2>0$, and $h_3<0$ satisfying the condition $\sum_{y=1}^{3} h_y p(y\mid \delta=1,z)=0$ for all $z\in\{1,\dots, q\}$. Employing Lemma \ref{appendix_lem} with $a=h_1$, $b=h_2$, $c=-h_3$, and $r_1=r_2=1$, we derive:

\begin{align*}
\begin{split}
\frac{1}{\pi^{(1)}_1} - \frac{1}{\pi^{(2)}_1} = h_1, \quad
\frac{1}{\pi^{(1)}_2} - \frac{1}{\pi^{(2)}_2} = h_2, \quad
\frac{1}{\pi^{(1)}_3} - \frac{1}{\pi^{(2)}_3} = h_3,
\end{split}
\end{align*}
where $\sum_{j=1}^3 \pi^{(1)}_j =\sum_{j=1}^3 \pi^{(2)}_j= 1$. Substituting $h_1$, $h_2$, and $h_3$ into $\sum_{y=1}^{3} h_y p(y\mid \delta=1,z)=0$ shows that the model is unidentifiable.

Lastly, we consider the case of $p\geq 4$. Suppose $h_y\,(y=1,\dots,p)$ satisfies \eqref{complete_cond}. Within $(h_1,\cdots,h_{p})$, we select three elements with signs as positive, positive, and negative, respectively, and define them as $ a$, $ b$, and $- c$ where $a,b,c>0$, and $\lambda$ is set to be sufficiently large to ensure that 
\begin{align}
\lambda > 2 \max \left\{ \frac{a+b}{ab},~\frac{1}{c} \right\}\label{cond_lam}.
\end{align}
For ease of notation, we denote $(h_1,\cdots,h_{p})=(h_1,\cdots,h_{p-3},a, b,- c)$. The remaining part of the proof is similar when the combination of the signs is negative, negative, and positive.
With the selected $\lambda$, $0<\pi_y^{(k)}<1\,(y=1,\ldots, p-3; k=1,2)$ are determined to be sufficiently small to satisfy
\begin{align}
&\left( 1- \sum_{y=1}^{p-3}\pi_y^{(1)} \right) \left( 1- \sum_{y=1}^{p-3}\pi_y^{(2)} \right)\geq \frac{1}{2},\quad \sum_{y=1}^{p-3}\pi_y^{(1)}<1, \quad  \sum_{y=1}^{p-3}\pi_y^{(2)}<1,\label{cond_r1r2}\\
&\frac{1}{\pi^{(1)}_y} - \frac{1}{\pi^{(2)}_y} = \lambda h_y,\quad \mathrm{for}~ y=1,\ldots,p-3\notag.
\end{align}
Furthermore, we define $r_1$ and $r_2$ as
\begin{align}
r_1^2=1- \sum_{y=1}^{p-3}\pi_y^{(1)},\ r_2^2=1- \sum_{y=1}^{p-3}\pi_y^{(2)}\label{def_r1r2}.
\end{align}

By determining the variables through these steps, it follows from \eqref{cond_lam}, \eqref{cond_r1r2}, and \eqref{def_r1r2} that condition \eqref{cond_abc} with $a=\lambda a$, $b=\lambda b$, and $c=\lambda c$ is fulfilled:
\begin{align*}
    &\frac{r_1^2 - r_2^2}{r_1^2r_2^2}\leq 2(r_1^2 - r_2^2)\leq \frac{2}{c} c<(\lambda c),\\
    & - \frac{ (\lambda a)(\lambda b)}{(\lambda a)+(\lambda b)}<
   - \frac{ab}{a+b}\frac{2(a+b)}{ab}=-2r_1^2r_2^2 \frac{1}{r_1^2r_2^2}\leq -\frac{1}{r_1^2r_2^2}    
    <\frac{r_1^2 - r_2^2}{r_1^2r_2^2}.
\end{align*}
Therefore, by applying Lemma \ref{appendix_lem}, we demonstrate that there exist $\pi^{(k)}_{p-2}$, $\pi^{(k)}_{p-1}$, and $\pi^{(k)}_{p}\,(k=1,2)$ such that
\begin{align*}
&\sum_{y=p-2}^{p}\pi_y^{(1)}=r_1^2,\ \sum_{y=p-2}^{p}\pi_y^{(2)}=r_2^2, \\
&\frac{1}{\pi_{p-2}^{(1)}}-\frac{1}{\pi_{p-2}^{(2)}}=\lambda a,\quad  \frac{1}{\pi_{p-1}^{(1)}}-\frac{1}{\pi_{p-1}^{(2)}}=\lambda b,\quad  \frac{1}{\pi_p^{(1)}}-\frac{1}{\pi_p^{(2)}}=-\lambda c.
\end{align*}
The condition \eqref{complete_cond} suggests that the constructed $\pi^{(k)}_y\,(y=1,\dots, p;k=1,2)$ satisfy $\sum_{y=1}^p \pi^{(k)}_y=1$ for $k=1,2$ and, for any $z\in\{1,\dots, q\}$,
$$\sum_{y=1}^{p}\lp \frac{1}{\pi^{(1)}_y}- \frac{1}{\pi^{(2)}_y}\rp p(y\mid \delta=1,z) = \lambda \sum_{y=1}^p h_y  p(y\mid \delta=1,z)=0.$$
Therefore, the model is unidentifiable.

\end{proof}

\begin{proof}[Proof of Theorem \ref{main}]
We consider when $y$ is continuous because when $y$ is discrete, we just need to change the integral to summation.
To simplify the discussion, we consider the case where $\mathcal{S}_y=\mathbb{R}$.
Let $\bm{u}$ be a fixed value.
Because $h$ and $g$ are injective functions, it is sufficient to prove the case where $\alpha:=h(\bm{u};\bm{\alpha})$ and $\beta:=g(\bm{u};\bm{\beta})$. Therefore, our goal is to prove
\begin{align*}
    \frac{ p(y\mid \bm{x}, \delta=1;\bm{\gamma}) }{\int p(y\mid \bm{x}, \delta=1;\bm{\gamma})  \Psi\{ \alpha + \beta m(y) \} ^{-1} dy }
    =\frac{ p(y\mid \bm{x}, \delta=1;\bm{\gamma^{\prime}}) }{\int p(y\mid \bm{x}, \delta=1;\bm{\gamma^{\prime}})  \Psi\{ \alpha^{\prime} + \beta^{\prime} m(y) \} ^{-1} dy },
\end{align*}
implies $\alpha=\alpha^{\prime}$, $\beta=\beta^{\prime}$ and $\bm{\gamma}=\bm{\gamma^{\prime}}$. Integrating both sides of the above equation with respect to $y$ yields the equality of the denominator. Thus, we have $p(y\mid \bm{x}, \delta=1;\bm{\gamma})=p(y\mid \bm{x}, \delta=1;\bm{\gamma^{\prime}})$; this implies $\bm{\gamma}=\bm{\gamma^{\prime}}$ by (C4).

Next, we consider the identification of $\beta$. Taking $\bm{z}_1$ and $\bm{z}_2$ such that they satisfy (C5), we show that
\begin{align}
   \int \frac{p(y\mid \bm{u}, \bm{z}_1, \delta=1;\bm{\gamma})}{\Psi\{ \alpha + \beta m(y) \}} dy&=\int \frac{p(y\mid \bm{u}, \bm{z}_1, \delta=1;\bm{\gamma})}{\Psi\{ \alpha^{\prime} + \beta^{\prime} m(y) \}} dy, \label{A1}\\
   \int \frac{p(y\mid \bm{u}, \bm{z}_2, \delta=1;\bm{\gamma})}{\Psi\{ \alpha + \beta m(y) \}} dy&=\int \frac{p(y\mid \bm{u}, \bm{z}_2, \delta=1;\bm{\gamma})}{\Psi\{ \alpha^{\prime} + \beta^{\prime} m(y) \}} dy,\label{A2}
\end{align}
implies $\beta=\beta^{\prime}$. It follows from (\ref{A1}) and (\ref{A2}) that
\begin{align}
    &\int K(y;\alpha,\alpha^{\prime},\beta, \beta^{\prime}) p(y\mid \bm{u}, \bm{z}_1, \delta=1;\bm{\gamma}) dy \notag\\
    &=\int K(y;\alpha,\alpha^{\prime},\beta, \beta^{\prime}) p(y\mid \bm{u}, \bm{z}_2, \delta=1;\bm{\gamma}) dy=0,\label{A3}
\end{align}
where $ K(y;\alpha,\alpha^{\prime},\beta, \beta^{\prime})=\Psi^{-1}\{ \alpha + \beta m(y) \} -\Psi^{-1}\{ \alpha^{\prime} + \beta^{\prime} m(y) \}$.
It remains to show that (\ref{A3}) implies $\beta=\beta^{\prime}$ in the following two steps:

\hspace{-1em}Step $\mathrm{I}$. We prove that the function $K(y;\alpha,\alpha^{\prime},\beta, \beta^{\prime})$ has a single change of sign when $\beta\neq \beta^{\prime}$.
Assume that $\beta\neq \beta^{\prime}$. The equation $K(y;\alpha,\alpha^{\prime},\beta, \beta^{\prime})=0$ has only one solution $y^*\in\mathcal{S}_y$ satisfying $m(y^*)=(\alpha-\alpha^{\prime})/(\beta^{\prime}-\beta)$ because of the injectivity of the function $m(\cdot)$ and $\Psi(\cdot)$. This implies $K(y)$ has a single change of sign.

\hspace{-1em}Step $\mathrm{I}\hspace{-1.2pt}\mathrm{I}$. We prove that the equation (\ref{A3}) does not hold when $\beta=\beta^{\prime}$. Without loss of generality, by Step $\mathrm{I}$, we consider a case where $K(y;\alpha,\alpha^{\prime},\beta, \beta^{\prime})<0\ (y<y^*)$ and $K(y;\alpha,\alpha^{\prime},\beta, \beta^{\prime})>0\ (y>y^*)$, and $p(y\mid \bm{u},\bm{z}_2, \delta=1)/ p(y\mid \bm{u},\bm{z}_1, \delta=1)$ is monotone increasing. Let $c$ be the upper bound of the density ratio
\begin{align*}
    c:=\sup_{y<y^*} \frac{  p(y\mid \bm{u},\bm{z}_2, \delta=1) }{p(y\mid \bm{u},\bm{z}_1, \delta=1)}.
\end{align*}
By a property on $K(y;\alpha,\alpha^{\prime},\beta, \beta^{\prime})$ shown in (\ref{A3}), we have
\begin{align*}
    0&=\int K(y;\alpha,\alpha^{\prime},\beta, \beta^{\prime}) p(y\mid \bm{u}, \bm{z}_2, \delta=1) dy\\
    &=\int_{-\infty}^{y^*} K(y;\alpha,\alpha^{\prime},\beta, \beta^{\prime}) \frac{p(y\mid \bm{u}, \bm{z}_2, \delta=1)}{p(y\mid \bm{u},\bm{z}_1, \delta=1)}p(y\mid \bm{u},\bm{z}_1, \delta=1) dy\\
    &\hspace{2em}+\int_{y^*}^{\infty} K(y;\alpha,\alpha^{\prime},\beta, \beta^{\prime}) \frac{p(y\mid \bm{u}, \bm{z}_2, \delta=1)}{p(y\mid \bm{u},\bm{z}_1, \delta=1)}p(y\mid \bm{u},\bm{z}_1, \delta=1) dy\\
    &\geq \int_{-\infty}^{y^*} cK(y;\alpha,\alpha^{\prime},\beta, \beta^{\prime}) p(y\mid \bm{u},\bm{z}_1, \delta=1) dy+\int_{y^*}^{\infty}c K(y;\alpha,\alpha^{\prime},\beta, \beta^{\prime})p(y\mid \bm{u},\bm{z}_1, \delta=1) dy\\
    &=c\int K(y;\alpha,\alpha^{\prime},\beta, \beta^{\prime}) p(y\mid \bm{u}, \bm{z}_1, \delta=1) dy=0,
\end{align*}
where the inequality follows from the definition of $c$. This results in the density ratio $p(y\mid \bm{u},\bm{z}_2, \delta=1)/ p(y\mid \bm{u},\bm{z}_1, \delta=1)$ being a constant on $\mathcal{S}_y$, hence, $p(y\mid \bm{u},\bm{z}_2, \delta=1)= p(y\mid \bm{u},\bm{z}_1, \delta=1)$ on $\mathcal{S}_y$. This contradicts with (C5), thus $\beta=\beta^{\prime}$. 

Finally, from the strict monotonicity of $\Psi$, it follows that the integration
\begin{align*}
    \int \frac{p(y\mid \bm{u}, \bm{z}_1, \delta=1;\bm{\gamma})}{\Psi\{ \alpha + \beta m(y) \}} dy,
\end{align*}
is injective with respect to $\alpha$. Therefore, equation (\ref{A1}) implies that $\alpha=\alpha^{\prime}$.

\end{proof}

\begin{proof}[Proof of Proposition \ref{cor_normal}]
It follows from the assumption (\ref{additional condition}) that there exist $M,C>0$ such that
\begin{align*}
&\int \frac{ p(y\mid \bm{x}, \delta=1;\bm{\gamma})}{ \Psi\{ h(\bm{u};\bm{\alpha}) + g(\bm{u};\bm{\beta}) m(y) \}} dy\\
&\propto \int_{-\infty}^{\infty} \exp \left\{ -\frac{1}{2}\frac{(y-h(\bm{u};\bm{\alpha})-\beta\mu(\bm{x},\bm{\kappa}) )^2}{\beta^2\sigma^2} \right\} \frac{1}{\Psi(y)\exp(|y|^s)}\exp(|y|^s) dy\\
&\leq \int_{-\infty}^{-M} \exp \left\{ -\frac{1}{2}\frac{(y-h(\bm{u};\bm{\alpha})-\beta\mu(\bm{x},\bm{\kappa}) )^2}{\beta^2\sigma^2} \right\} C\exp(|y|^s) dy +C<\infty,
\end{align*}
where $0<s<2$.
The first and the second terms of the last equation hold by the condition (\ref{additional condition}) and the increasing assumption of $\Psi$, respectively.
\end{proof}

\end{document}